\documentclass[letterpaper, 10 pt, conference]{ieeeconf}   

\IEEEoverridecommandlockouts                              

\usepackage{caption}
\usepackage{mathtools}
\usepackage{nomencl}
\usepackage{cite}
\usepackage{ntheorem}   
\usepackage{lipsum}     
\usepackage{graphicx}
\usepackage{amsmath}
\usepackage{algpseudocode}
\usepackage{algorithm}

\usepackage{flushend}

\usepackage{esint}
\usepackage[latin9]{luainputenc}
\usepackage{amssymb}
\usepackage{lineno,hyperref}
\usepackage{xcolor}
\modulolinenumbers[5]
\usepackage{bm}

\DeclareRobustCommand{\uvec}[1]{{%
		\ifcsname uvec#1\endcsname
		\csname uvec#1\endcsname
		\else
		\bm{\mathbf{#1}}%
		\fi
}}

\theoremstyle{plain}
\newtheorem{prop}{Proposition}

\newtheorem{rem}{Remark}

\graphicspath{{./graphics/}}

\date{March 2021}
\renewcommand{\baselinestretch}{.97}
\newcommand{\makehighlight}[1]{\textcolor{black}{#1}}

\begin{document}


\title{Discretization and Stabilization of Energy-Based Controller for Period Switching Control and Flexible Scheduling}

\author{%
Seyed Amir Tafrishi$^{1}$, Xiaotian Dai$^{2}$, Yasuhisa Hirata$^{1}$ and Alan Burns$^{2}$
\thanks{$^{1}$ S. A. Tafrishi and Y. Hirata are with Department of Robotics, Tohoku University, 6-6-01 Aramaki-Aoba, Aoba-ku, Sendai 980-8579, Japan.
        {\tt\small \{s.a.tafrishi, hirata\}@srd.mech.tohoku.ac.jp}}%
\thanks{$^{2}$ X. Dai and A. Burns are with the Department of Computer Science, University of York, United Kingdom.
        {\tt\small \{xiaotian.dai, alan.burns\}@york.ac.uk}}%
}

\maketitle

\begin{abstract}
Emerging advanced control applications, with increased complexity in software but limited computing resources, suggest that real-time controllers should have adaptable designs. These control strategies also should be designed with consideration of the run-time behavior of the system. One of such research attempts is to design the controller along with the task scheduler, known as control-scheduling co-design, for more predictable timing behavior as well as surviving system overloads. Unlike traditional controller designs, which have equal-distance sampling periods, the co-design approach increases the system flexibility and resilience by explicitly considering timing properties, for example using an event-based controller or with multiple sampling times (non-uniform sampling and control). 
Within this context, we introduce the first work on the discretization of an energy-based controller that can switch arbitrarily between multiple periods and adjust the control parameters accordingly without destabilizing the system. A digital controller design based on this paradigm for a DC motor with an elastic load as an example is introduced and the stability condition is given based on the proposed Lyapunov function. The method is evaluated with various computer-based simulations which demonstrate its effectiveness.  
\end{abstract}

\section{Introduction} \label{sec:introduction}
\makehighlight{Utilising computational resources efficiently is important for emerging systems that have complex functionalities but with limited on-board computing resources, for example, autonomous mobile robots and self-driving vehicles.}
In the control systems community, traditional controller design rarely takes computing resources into consideration. For instance, in computer-based control, the discrete controller is often designed with the assumption that the control task, i.e. the program that runs the controller code on an embedded computer, is executed with equal-distance sampling and control intervals. Also, the discrete controller is directly derived based on the period from the continuous controller design. These, however, (a) limit the flexibility of software and scheduling; and (b) degrade the control performance if the actual sampling/control intervals are different.

From the task scheduling perspective, historically, tasks are scheduled with more deterministic policies, e.g. \emph{cyclic executives} \cite{burns2001real, burns2015cyclic}. \makehighlight{These policies have less flexibility but with more timing predictability that the real-time controller can be executed with a relatively stable sampling time and jitter. However, as the computing architecture becomes more complicated and demanding for more functionalities, the execution environment soon becomes more dynamic. This issue brings both challenges and opportunities for the design of digital controllers.}
In recent years, there is ongoing research to support \emph{control scheduling co-design} \cite{dai2019dual, arzen2000introduction}, in which the computation resources are taken into consideration explicitly during the design of the real-time controller.
Pioneering works have shown that event-based control is comparable to periodic control in \cite{astrom2008event, lehmann2011event, wang2017dynamic, lunze2010state} that can significantly save computing and energy resources.

Recently, it has been realized in both the control and scheduling communities that real-time controllers can run at different frequencies (and hence so does the control task that executes the control-related functions) \cite{cervin2000feedback, dai2020period}. This brings opportunities for real-time embedded control systems to adapt the control period in accordance with the current system workload to satisfy task schedulability and avoid overloading the system. Note that we assume sampling time (the time interval in which the controller updates its current states; which is not necessarily equal to the actual sampling rate of the sensor) equals the control period; hence, we use \emph{control period} and \emph{sampling period} interchangeably. We also use the term \emph{interval} and \emph{period} interchangeably as the semantics are identical in this paper.

Unlike the more traditional mode switching control, period switching focuses on adapting different control periods, rather than switching between a set of controllers with pre-designed tuned parameters. As the task utilization is equal to its worst-case execution time divided by its inter-release interval (task period), adapting the interval can have a significant impact on computing resource usage and the overall system's energy consumption.

\makehighlight{ On the other hand, energy-based control \cite{fantoni2000energy,dahleh2004lectures,9382723} is not explored as much as \emph{proportional-integral-derivative} (PID) control \cite{aastrom2001future,johnson2005pid} or state feedback control \cite{4659500,sevim2016stabilization}. With physical systems, energy-based control is often more intuitive as it explores the internal dynamics of the physical process. It supports non-linear systems as well as underactuated systems \cite{847110} and has large regional convergence compared to controllers designed with pole-placement. It can naturally support non-uniform sampling and it should be easier to stabilize in this case. We think that because the convergence of the physical system can be derived with dependency on sampling time, it should be easier to develop stable control strategies for these time-varying systems. Also, this fact was confirmed in insightful studies of motor systems with flexible loads \cite{spong1996energy,gao2009energy,ujjal2021WAC} where an energy-based control strategy was able to stabilize the whole system under physical/virtual constraints.}


To fulfill the requirement of control-scheduling co-design, a vital step is to discretize the real-time controller, which is to transform the system dynamics from continuous-time to discrete-time as a function of the sampling period. 
Controllers with non-uniform sampling can be important in physical systems \cite{sevim2016stabilization}, and developing a discrete-time controller can improve the response, save energy and minimise the anomalies when sampling frequencies are altered.
Historically, the energy-based controller has largely been addressed in the continuous-time domain while the transformation into discrete-time is not studied that much. \makehighlight{For example, pioneering research was developed as the first discretized energy-based controller for an underactuated system with variable disturbances \cite{franco2018discrete}.} 
In our work, we leverage the control-scheduling co-design by providing discretization of energy-based control, that can be used for the purpose of robust, adaptive and flexible task scheduling \cite{dai2018flexible}.

This paper is organized as follows: we develop the DC motor model with elastic load and its discretization in Section~\ref{sec:controller}. Next, we introduce the design of a discrete energy-based controller in Section~\ref{sec:switch} and derive its stability conditions. An evaluation with simulation is given in Section~\ref{sec:evaluation}, followed by concluding remarks and a discussion of future work in Section~\ref{sec:conclusion}.

\section{Energy-based Controller Design} \label{sec:controller}
In this section, we first present the model of the DC motor with an elastic load. Then, we demonstrate the discretized model with its general relations due to the Maclaurin series approximation.

\subsection{Model of a DC Motor System}
We consider a DC motor model with attached elastic load, stiffness $K_L$ and viscous friction $B_L$, as shown in Fig.~\ref{Fig:Motor_DC}. 
The equations of motion for this DC motor are \cite{ruderman2008optimal,Ujjal2021,ujjal2021WAC}:
\begin{align}
\label{Eq:StateEquationsMot}
&\dot{\uvec{x}}=  \uvec{A} \uvec{x}+ \uvec{B} u, \\
&\uvec{y}=\uvec{C}\uvec{x},
\end{align}
where 
\begin{equation}
\begin{split}
\uvec{A}=\left[\begin{array}{ccc}
-\frac{R}{L} & -\frac{K_b}{L} & 0 \\
\frac{K_m}{J}  & -\frac{B}{J} & -\frac{K_L}{J}\\
0 & 0 & 1 
\end{array}\right], \uvec{B}=\left[\begin{array}{c}
-\frac{1}{L}  \\
0\\
0
\end{array}\right], 
\uvec{C}=\left[\begin{array}{c}
1  \\
1 \\
1
\end{array}\right] 
\end{split}
\label{Eq:Dynamicmotormodel}
\end{equation}
where $\uvec{x}=[I\;\; \dot{\theta}\;\;\theta]^T$, $u$, $L$, $R$, $K_m$ and \makehighlight{$K_b$} are the DC motor current and angular velocity and displacement states, the input voltage as the control input, the armature inductance and the armature resistance, back emf constant \makehighlight{and motor viscous friction}, respectively. Also, the system's inertia $J=J_m+J_L$ encompasses the motor $J_m$ and load $J_L$ inertia and viscous friction $B=B_m+B_L$ are motor $B_m$ and load $B_L$ viscous frictions. 

\begin{figure}[t!]
    \centering
    \includegraphics[width=2.2 in]{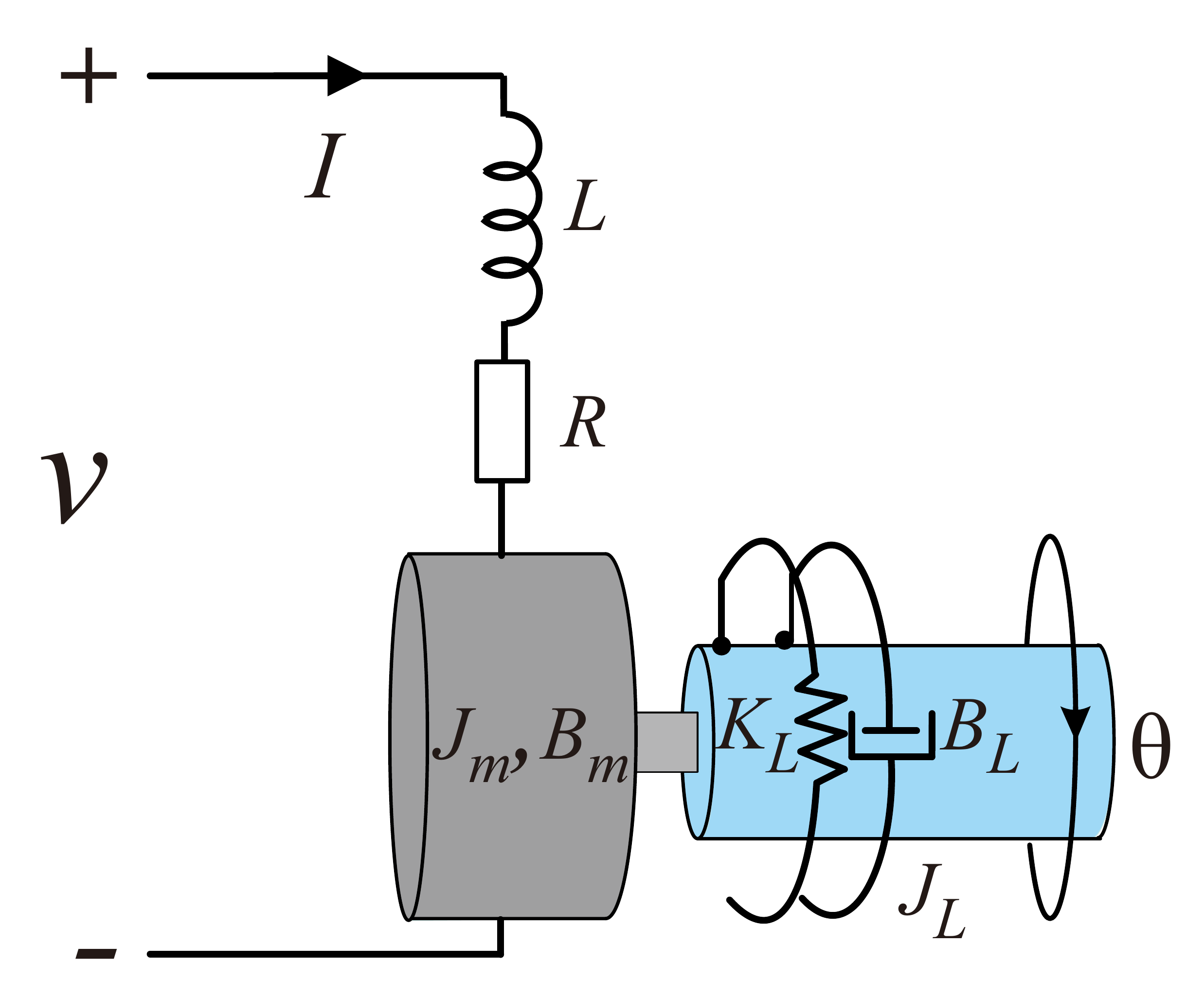}
	\caption{The DC motor model with flexible load.} 
	\label{Fig:Motor_DC}
\end{figure}

\subsection{Discrete State-Space Model}
The discretized model of (\ref{Eq:Dynamicmotormodel}) is presented as \cite{van1978computing} 
\begin{equation}
    \uvec{x}_{k+1}=\uvec{F}_k\uvec{x}_k+\uvec{G}_k \uvec{u}_k,
    \label{Eq:Discretemodelstate}
\end{equation}
where
\begin{equation}
    \uvec{F}_k\coloneqq \bm{e}^{\uvec{A}h_k},   \;\; \uvec{G}_k\coloneqq  \sideset{}{_{0}^{h_k}} \intop \bm{e}^{\uvec{A}h_k}\;  \uvec{B} \; d \tau   , 
    \label{Eq:DiscereteFkGK}
\end{equation}
where $h_k=t_k-t_{k-1}$ is the sampling time at $k$-th sample. 
The reason to use an exponential model rather than $z$-transformation is to facilitate the inclusion of the sampling time as a variable in the time domain.
Then, we define exponential matrix operations for (\ref{Eq:Discretemodelstate})-(\ref{Eq:DiscereteFkGK}) by using the Maclaurin series as follows
\begin{equation}
 \bm{\Phi}(\uvec{A}) \triangleq \sum^{\infty}_{i=0} \frac{\uvec{A}^i}{(i+1)!}= \uvec{I}+\frac{\uvec{A}}{2!}+\frac{\uvec{A}^2}{3!}+...
 \label{Eq:maclurtianexponetialexp}
\end{equation}
where we can have the following equality relations derived from (\ref{Eq:maclurtianexponetialexp})
\begin{align}
\label{Eq:Relationmultipli}
&(i)\;\;\;\;\uvec{A} \bm{\Phi}(\uvec{A}) = \bm{\Phi}(\uvec{A})\uvec{A}, \\
&(ii)\;\;\;\;\bm{e}^{\uvec{A}}=\uvec{I}+\uvec{A}\bm{\Phi}(\uvec{A}),\\
&(iii)\sideset{}{_{0}^{h}} \intop \bm{e}^{\uvec{A}h} d \tau = h \bm{\Phi}(\uvec{A}h), \;\;\;\; \\
& (iv) \; \bm{\Phi}(\uvec{T}^{-1}\uvec{A}\uvec{T})=\uvec{T}^{-1}\bm{\Phi}(\uvec{A})\uvec{T}\; \textnormal{for an arbitrary} \nonumber \\
&\textnormal{invertible\;\uvec{T}},
\label{Eq:Relationmultilast}
\end{align}
where $\uvec{A}\in\mathbb{R}^{n\times n}$ and $h\in\mathbb{R}$ are arbitrary variables in here.
The proof of Eqs. (\ref{Eq:Relationmultipli})-(\ref{Eq:Relationmultilast}) can be found in \cite{sevim2016stabilization}.

\section{Arbitrary Switching with Energy-based Controller} \label{sec:switch}
In this section, we derive a discrete energy-based controller for an arbitrary switching mode.
Also, \makehighlight{we explain how certain controller gains get updated based on arbitrary sampling time to improve the convergence while the system is maintained stable.}

In this control problem, we assume that system (\ref{Eq:Discretemodelstate}) with states $\uvec{x}_k$ is controlled to converge toward the desired states $\uvec{x}_d$. \makehighlight{The system (\ref{Eq:Discretemodelstate}) is assumed to have a random sampling time $h_k$ that changes continuously.} The aim is to keep the system stable during the converges with changing sampling time $h_k$. Now, we develop the following proposition for our energy-based controller with stable converges as follows:

\begin{prop}
Assume the discrete system defined by Eqs.~(\ref{Eq:Discretemodelstate})-(\ref{Eq:DiscereteFkGK}) has a random sampling time $h_k$ in $k$-th sample. Then, the physical system converges to the desired states $\uvec{x}_d$ with the following input
\makehighlight{
\begin{align}
    &u_k= -\left[ k_E\left(\frac{1}{2}\uvec{x}^T_k\uvec{D}\uvec{x}_k \right)\uvec{x}^T_k\uvec{D}\bm{\Phi}(\uvec{A}h_k)\uvec{B} \right]^{-1} \nonumber\\
    &\cdot \Bigg[\Bigg(k_E\left(\frac{1}{2}\uvec{x}^T_k\uvec{D}\uvec{x}_k \right)\uvec{x}^T_k\uvec{D} \Bigg)\bm{\Phi}(\uvec{A h_k})\uvec{A}\uvec{x}_k  \nonumber\\
    &  +\frac{k_D}{h_k}\left(\dot{\theta}_k-\dot{\theta}_d\right) \left( \uvec{F}_m \uvec{x}_k -\dot{\theta}_k\right) +k_P\left( \theta_k-\theta_d\right)\dot{\theta}_k\Bigg],
    \label{Eq:ControllerinputEPropo}
\end{align}}
while the following condition is satisfied for stability\makehighlight{
\begin{align}
    &k_E \left(\frac{1}{2}\uvec{x}^T_k\uvec{D}\uvec{x}_k \right)  \uvec{x}_k^T\uvec{D} \bm{\Phi}(\uvec{A} h_k)\uvec{B}u_k \nonumber \\
    &+k_P\left( \theta_k-\theta_d\right)\dot{\theta}_k \leq - k_E E_k \; \uvec{x}_k^T\uvec{D} \bm{\Phi}(\uvec{A}h_k)\uvec{A}\uvec{x}_k   \nonumber\\
    &-\frac{k_D}{h_k}\left(\dot{\theta}_k-\dot{\theta}_d\right) \left(\uvec{F}_m \uvec{x}_k-\dot{\theta}_k\right),
    \label{Eq:DeriLyapunovreordered}
\end{align}
where $\uvec{F}_m$ is the sub-matrix of $\uvec{F}$ that belongs to the model of the rotational body (armature) of the DC motor and $k_E,k_P$ and $k_D$ are the controller gains with
\begin{equation*}
    \uvec{D}=\left[\begin{array}{ccc}
L &  0 &0\\
 0 & J &0 \\ 
 0 & 0 & K_L
\end{array}\right].
\end{equation*}
Also, the controller should stay stable with the inclusion of boundary conditions for the sampling time $h_k$
\begin{align}
\begin{cases}
   \begin{split}
    & 
    k_D\left(\dot{\theta}_k-\dot{\theta}_d\right)  \left(\uvec{F}_m \uvec{x}_k-\dot{\theta}_k\right)\leq 0    \end{split} \hspace{.56 in} h_k\rightarrow0 \\
    &  - \uvec{A}\uvec{x}_k \leq \uvec{B}u_k \hspace{1.64 in} h_k \rightarrow \infty \\
    \end{cases}
     \label{Eq:DeriLyapunovreorderedlowhintext}
\end{align}}

\end{prop}
\begin{proof}
The energy equation consisting the kinetic and potential energies, for the DC motor with elastic load system is as follows
\begin{equation}
   E=\frac{1}{2}\uvec{x}^T \uvec{D}\uvec{x}.
    \label{Eq:Energyequation1}
\end{equation}
Then, the derivative of energy (\ref{Eq:Energyequation1}) in the time domain is given by
\begin{eqnarray}
    \dot{E} =\uvec{x}^T \uvec{D}\dot{\uvec{x}}.
\end{eqnarray}
Now, the new discrete energy equation that encompasses the discrete state-space model using (\ref{Eq:Discretemodelstate})-(\ref{Eq:DiscereteFkGK}) with sampling period $h_k$ \makehighlight{is defined by the first-order forward difference of Taylor expansions \cite{strauss2007partial,meena2020discretization}}
\makehighlight{
\begin{align}
    &  E'=\frac{1}{h_k}\left(E_{k+1}-E_k \right) \coloneqq  \uvec{x}^T_{k}\uvec{D}\frac{1}{h_k}(\uvec{x}_{k+1}-\uvec{x}_{k}) \nonumber \\
    & =  \uvec{x}^T_k \uvec{D} \Big[  \bm{\Phi}( \uvec{A}h_k) \left( \uvec{A} \uvec{x}_k+ \uvec{B}u_k \right)  \Big].
    \label{Eq:derivativeenergy}
\end{align}}
Then, we propose a Lyapunov function candidate as follows
\begin{equation}
    V=\frac{1}{2}k_E{E}^2+\frac{1}{2}k_D(\dot{\theta}-\dot{\theta}_d)^2+\frac{1}{2}k_P(\theta-\theta_d)^2,
    \label{Eq:Lyapunovfunc}
\end{equation}
where $\{\theta_d,\dot{\theta}_d\}$ are the desired angular position and velocity of the motor shaft. In (\ref{Eq:Lyapunovfunc}), the first term is for converging the overall energy to zero while the other two terms converge the system to desired position and velocity. 
It is clear that the two last terms of the Lyapunov function are bounded based on the desired states and the energy term {\small $E^2$} is always upper bounded if the control input $u_k<u_{sat}$ is constrained with $u_{sat}$. 
Then, the derivative of the Lyapunov function (\ref{Eq:Lyapunovfunc}) is
\begin{equation}
    \dot{V}(t)=k_E E \dot{E} + k_D(\dot{\theta} - \dot{\theta}_d)\ddot{\theta} + k_P(\theta - \theta_d)\dot{\theta},
    \label{Eq:DeriLyapunovfunctime}
\end{equation}
Based on the time domain Lyapunov function (\ref{Eq:DeriLyapunovfunctime}), \makehighlight{we propose a discretized approximated function $V'$ by using first-order forward derivative for the $k$-th sample as follows}
\makehighlight{
\begin{align}
    &V'\approx \frac{1}{h_k}(V_{k+1}-V_{k})=  \frac{1}{h_k} k_E E_k \left(E_{k+1}-E_k\right) \nonumber\\
    &+ \frac{k_D}{h_k} \left(  \frac{1}{h_k} (\theta_{k+1}-\theta_k)-\dot{\theta}_d\right) \left(\dot{\theta}_{k+1}-\dot{\theta}_k\right) \nonumber\\ &+\frac{k_P}{h_k}\left(\theta_k-\theta_d \right)\left(\theta_{k+1}-\theta_k\right) ,
    \label{Eq:DeriLyapunovfunc}
\end{align}
}

\noindent where $E_k$ is the discretized formula of energy formula $E$ in (\ref{Eq:Energyequation1}). By substituting (\ref{Eq:Discretemodelstate}), (\ref{Eq:DiscereteFkGK}) and (\ref{Eq:derivativeenergy}), the re-ordered function in discrete domain, based on the (\ref{Eq:maclurtianexponetialexp})-(\ref{Eq:Relationmultilast}) relations and discrete velocity updates $\dot{\theta}_k=(1/h_k)(\theta_{k+1}-\theta_{k})$, can be derived as
\begin{align}
    &V'= k_E E_k \; \Big [\uvec{x}_k^T\uvec{D}  \bm{\Phi}(\uvec{A}h_k)\left(\uvec{A}\uvec{x}_k+\uvec{B}u_k \right) \Big ] \nonumber \\
    & +\frac{k_D}{h_k}\left(\dot{\theta}_k-\dot{\theta}_d\right) \left( \uvec{F}_m \uvec{x}_k-\dot{\theta}_k \right)+k_P\left( \theta_k-\theta_d\right)\dot{\theta}_k
    \label{Eq:DeriLyapunovfuncHybrid}
\end{align}

The new function candidate $V'$ should be stable around equilibrium with satisfying LaSalle's theorem. It can be seen, the function in Eq. (\ref{Eq:Lyapunovfunc}) is positive definite, thus, we can have following condition for stable convergence
\begin{equation}
  V'(\uvec{x}_k) \leq 0, \; \; \; \forall\uvec{x}_k \in \mathbb{R}^3
  \label{Eq:Lyapanouvconditionstab}
\end{equation}
To prove the stability condition (\ref{Eq:Lyapanouvconditionstab}) for (\ref{Eq:DeriLyapunovfuncHybrid}), we represent the following new terms $\uvec{A}^*=-\uvec{A}$ and $\uvec{F}^*_m=-\uvec{F}_m$. Then, the inequality (\ref{Eq:Lyapanouvconditionstab}) can be computed as follows
\begin{align}
    &k_E E_k \left[ \uvec{x}_k^T\uvec{D} \bm{\Phi}(\uvec{A}^*h_k)\uvec{B}u_k \right]   \nonumber \\
    &+k_P\left( \theta_k-\theta_d\right)\dot{\theta}_k \leq k_E E_k \; \uvec{x}_k^T\uvec{D} \bm{\Phi}(\uvec{A}^*h_k)\uvec{A}^*\uvec{x}_k   \nonumber\\
    &+\frac{k_D}{h_k}\left(\dot{\theta}_k-\dot{\theta}_d\right) \left(\uvec{F}^*_m \uvec{x}_k-\dot{\theta}_k\right)
    \label{Eq:DeriLyapunovreordered}
\end{align}
Inequality (\ref{Eq:DeriLyapunovreordered}) clarifies the constraints on control gains and sampling time $h_k$ in establishing a stable controller. \makehighlight{ Based on the transitivity and addition properties of the inequalities we can divide the inequality to two parts as:
\begin{align}
    &V_1:k_P\left( \theta_k-\theta_d\right)\dot{\theta}_k - \frac{k_D}{h_k}\left(\dot{\theta}_k-\dot{\theta}_d\right) \left(\uvec{F}^*_m \uvec{x}_k-\dot{\theta}_k\right) \leq 0 \nonumber \\
    &V_2:\uvec{B}u_k  +\uvec{A}\uvec{x}_k \leq 0
\end{align}
Because the system is always convergent to the desired states, $V_2$ is always satisfied. Also, $V_1$ with normalized parameters and constant $I$ and $\theta$ can be shown as Fig. \ref{Fig:3DSability1} where it is clear that system should stay in a stable region if $k_P \gg k_D$ chosen large values when $k_D<1$. Also, the system stability becomes harder in small sampling times $h_k$ when $k_P$ is not that high. }
 \begin{figure}[t!]
    \centering
    \vspace{3mm} 
    \includegraphics[width=1.9 in]{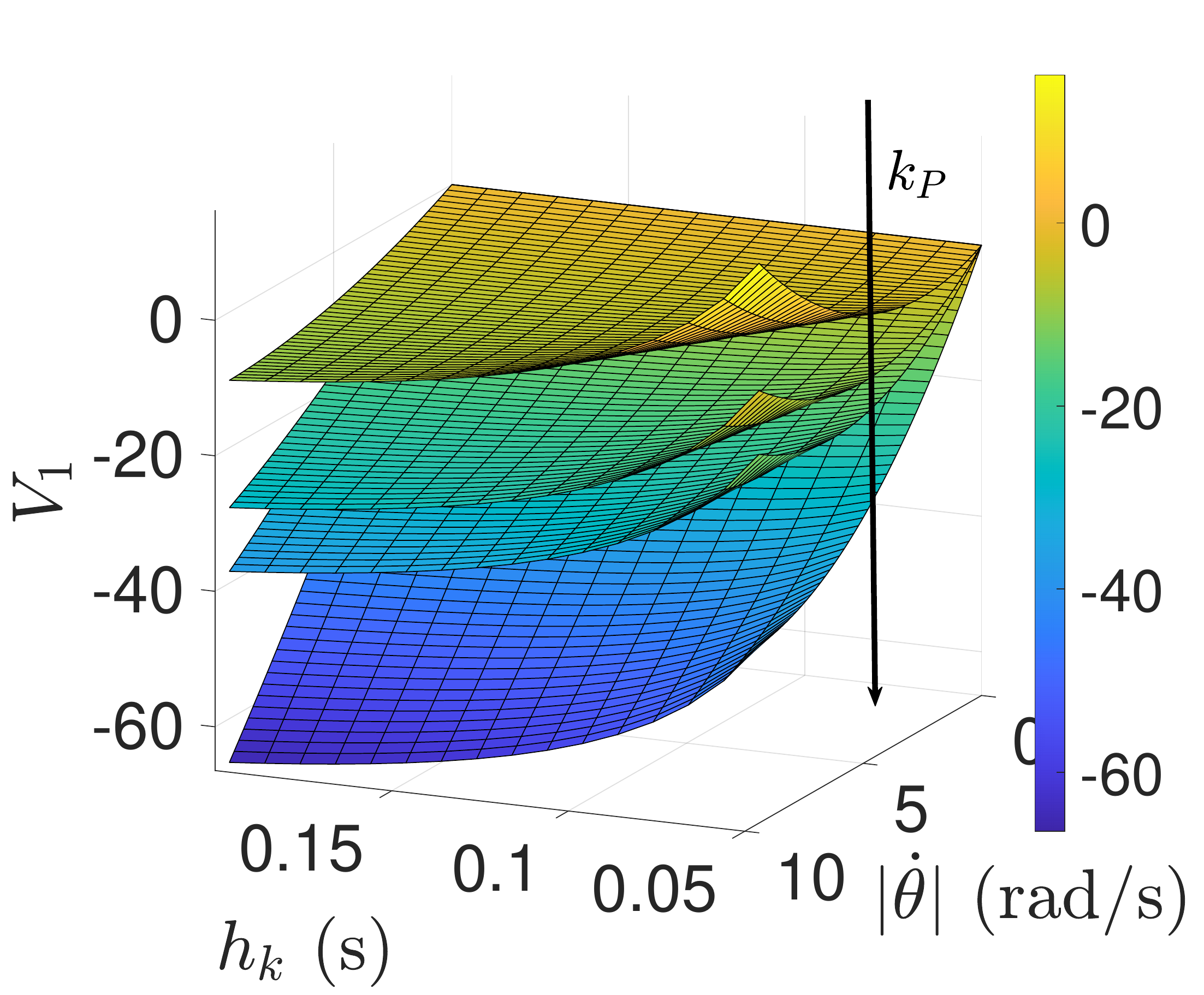}
	\caption{\makehighlight{Stability region variation with respect to sampling time $h_k$ and angular velocity $|\dot{\theta}|$, when $k_P$ is changing from a larger to smaller values.}}
	\label{Fig:3DSability1}
\end{figure}

The stability regarding the sampling time $h_k$ from (\ref{Eq:DeriLyapunovreordered}) can be simplified to two boundary cases of 
\begin{align}
\begin{cases}
   \begin{split}
    &0  \leq 
    k_D\left(\dot{\theta}_k-\dot{\theta}_d\right)  \left(\uvec{F}^*_m \uvec{x}_k-\dot{\theta}_k\right)  \end{split} \hspace{.8 in} h_k\rightarrow0 \\
    &\uvec{B}u_k  \leq \uvec{A}^*\uvec{x}_k \hspace{1.9 in} h_k \rightarrow \infty \\
    \end{cases}
     \label{Eq:DeriLyapunovreorderedlowhintext}
\end{align}
Thus, the controller gains should be satisfied by inequalities (\ref{Eq:DeriLyapunovreorderedlowhintext}) based on the range of change for the sampling time $h_k$ to always stay in the stable region. 
\makehighlight{Finally, by equating (\ref{Eq:DeriLyapunovfuncHybrid}) to zero, the energy-based control input (\ref{Eq:ControllerinputEPropo}) can be found.}
\end{proof}

\begin{rem} The developed controller input (\ref{Eq:ControllerinputEPropo}) has singularities which are located in two conditions 
\begin{align}
\label{Eq:singularfirstcondition}
    &h_k \leq \varepsilon_h\\
    &\frac{1}{2}\uvec{x}^T_k\uvec{D}\uvec{x}_k  =\left(J\dot{\theta}_k^2+K\theta^2_k+LI^2_k\right)\leq \varepsilon_c
    \label{Eq:singularsecondcondition}
\end{align}
\makehighlight{where $\varepsilon_h$ and $\varepsilon_c$ are constant small values. These conditions should always be satisfied to prevent the controller from becoming saturated. The condition (\ref{Eq:singularfirstcondition}) shows that the lower boundary of the sampling time is determined by $\varepsilon_h<h_k$ where $\varepsilon_h$ is a constant small value.} Also, if the sum of the squared angular displacement, velocity and motor current, as shown in the condition (\ref{Eq:singularsecondcondition}), becomes very small (near zero) where $\varepsilon_c$ is a very small constant value, there appear singularities. Also, (\ref{Eq:singularsecondcondition}) imposes another constraint to geometric parameters of the physical system i.e., $J$, $K$ and $L$, that cannot be lower than a certain value.
\end{rem}

To provide complete stability through mode switching in sampling time and desired states, aforementioned conditions should be satisfied and control variables should be designed accordingly. As we expect to have a random sampling time $h_k \in [h_{min}, h_{max}]$ the model (\ref{Eq:Discretemodelstate}) should always be stable; hence, the minimum $h_{min}$ and the maximum $h_{max}$ sampling time can be determined from (\ref{Eq:DeriLyapunovreorderedlowhintext}) (as shown in Fig. \ref{Fig:3DSability1}) based on the physical parameters of the considered model.

Next, we develop a tuning dynamic $k_E$ gain to keep the system in smooth and robust convergence while stability is satisfied with randomly changing sampling time. This issue is important because random sampling will act as disturbance to whole system dynamics and can create a divergence in the states. The first term of the Lyapunov function $k_E E_k E'(h_k)$ in (\ref{Eq:DeriLyapunovfunc}) is an important part of the controller that is a function of $h_k$. Also, as the sampling time $h_k$ variable changes, $E'$ function varies that can create instabilities in controller output $u_k$. To solve this issue, we assume that the controller is designed with the standard sampling time $h_s$
which results in a standard energy difference of 
\begin{equation}
E'_s=k_{E,s} E'(h_s),
\end{equation}
where $k_{E,s}$ is the standard constant energy gain for the standard sampling $h_s$. Then, by considering an arbitrary sampling $h_k$, $k_E$ is updated as follows:
\begin{equation}
    k_E(h_k)= E'_s/E'(h_k)+K_c,
    \label{Eq:SamplingKeTune}
\end{equation}
where $K_c$ is a constant small value. Eq.~(\ref{Eq:SamplingKeTune}) prevents the controller from producing large control values when the sampling time shifts suddenly with large deviations. \makehighlight{The stability of the newly developed dynamic gain can be found by substituting (\ref{Eq:SamplingKeTune}) to (\ref{Eq:DeriLyapunovreordered}) that results
\begin{align}
    &k_E \left(E'_s/E'(h_k)+K_c\right)  \left[ \uvec{x}_k^T\uvec{D} \bm{\Phi}(\uvec{A}^*h_k)\uvec{B}u_k \right]   \nonumber \\
    & \leq k_E \left(E'_s/E'(h_k)+K_c\right)  \; \uvec{x}_k^T\uvec{D} \bm{\Phi}(\uvec{A}^*h_k)\uvec{A}^*\uvec{x}_k   \nonumber\\
    &+\frac{k_D}{h_k}\left(\dot{\theta}_k-\dot{\theta}_d\right) \left(\uvec{F}^*_m \uvec{x}_k-\dot{\theta}_k\right),
    \label{Eq:DeriLyapunovreodynamicK}
\end{align}
which keeps the system in a stable region.  When the $k_D$ term becomes dominant in smaller values of $h_k$ as shown in (\ref{Eq:DeriLyapunovreorderedlowhintext}), the dynamic $k_E$ gets larger to create a trade off and keeping the condition (\ref{Eq:DeriLyapunovreodynamicK}) always valid.  }

 \begin{figure}[t!]
    \centering
    \vspace{3mm} 
    \includegraphics[width=2.8 in]{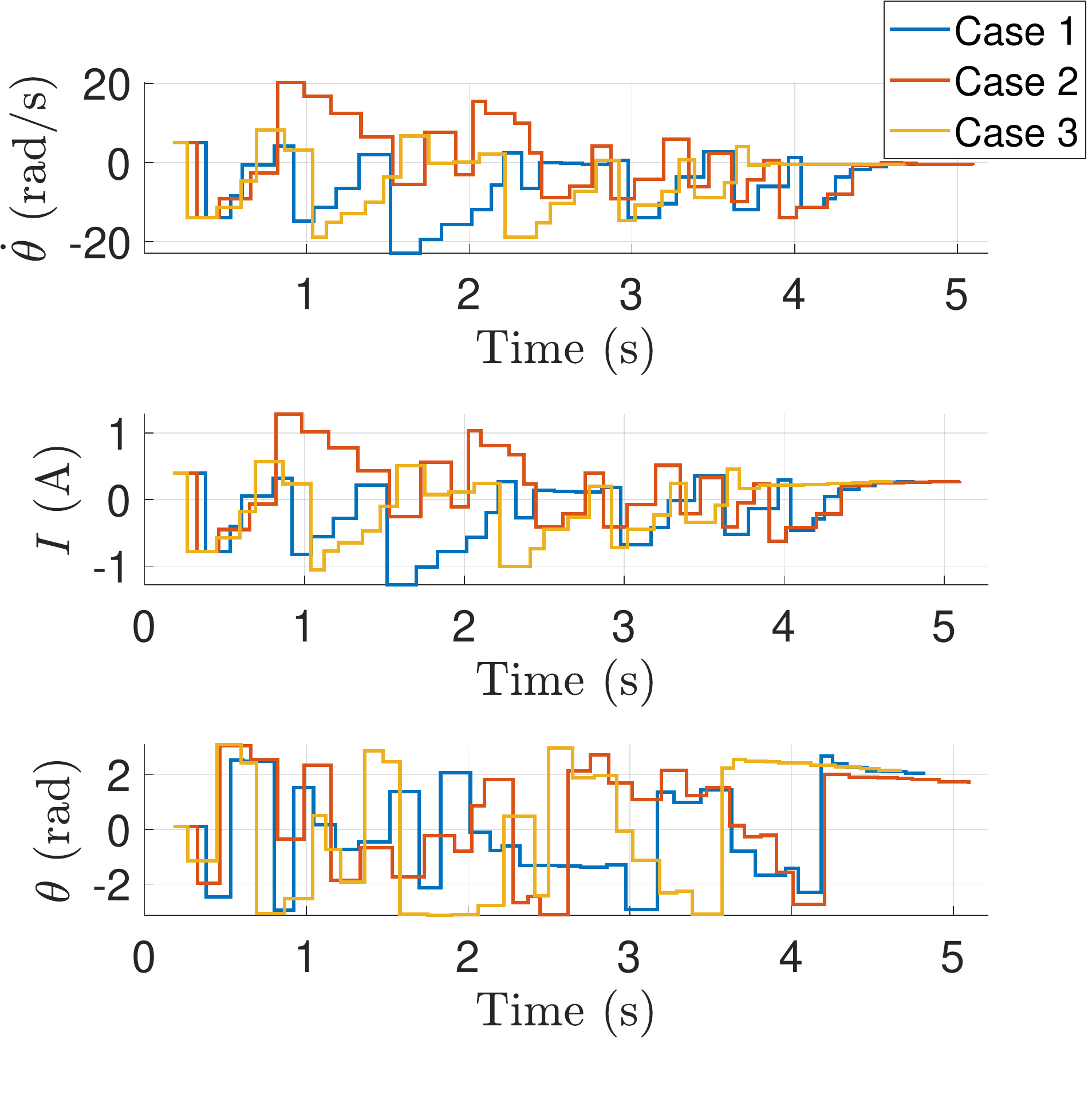}
	\caption{\makehighlight{Simulation results of the system states.}}
	\label{Fig:DiscreteResu1}
\end{figure}

\section{Evaluation} \label{sec:evaluation}
\begin{figure}[t!]
\centering
\includegraphics[width=2.8 in]{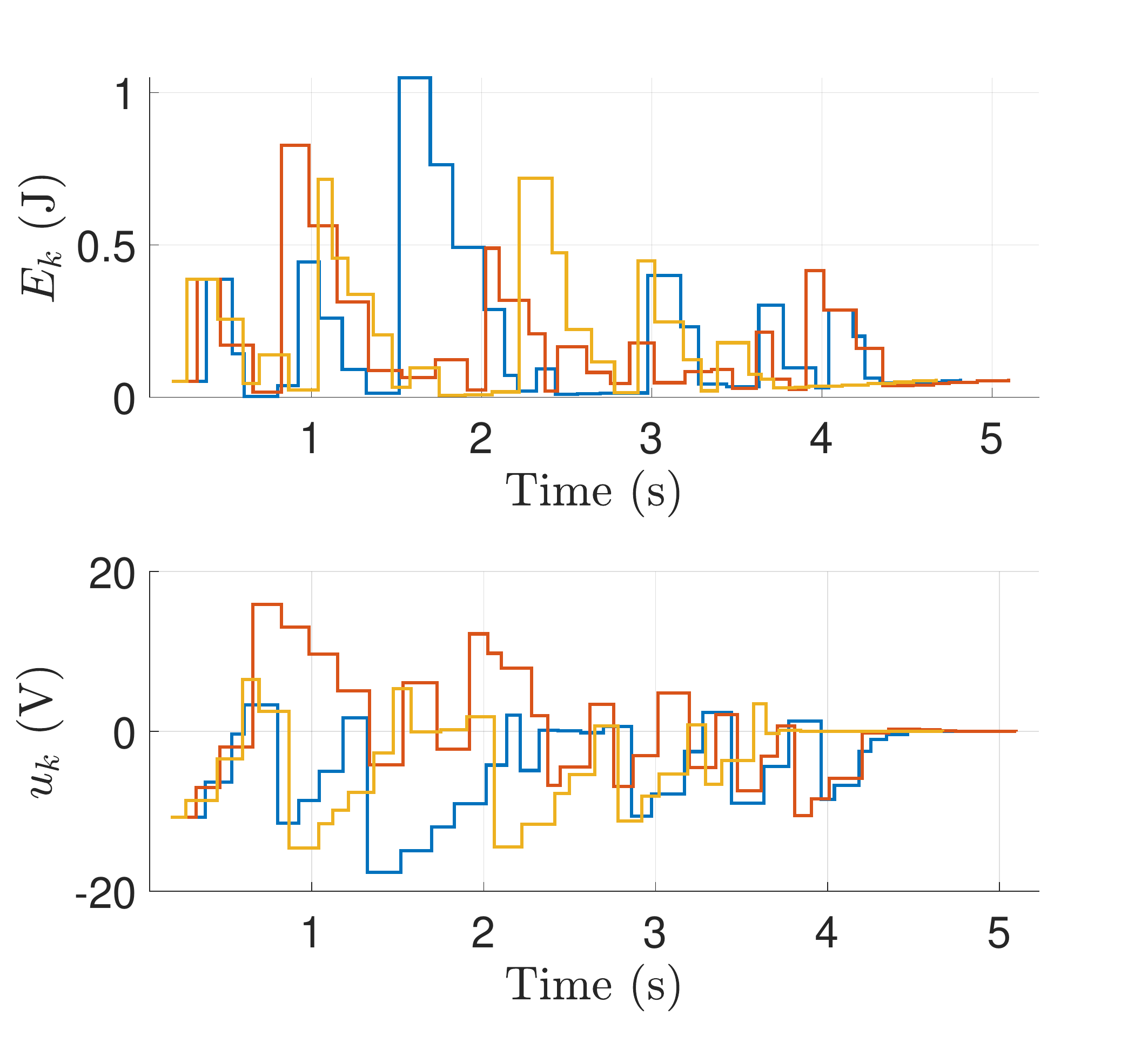}
\caption{\makehighlight{Overall energy (top) and controller inputs (bottom).}} 
\label{Fig:DiscreteResu2}
\end{figure}
In this section, we study the proposed controller's behavior for the DC motor model using simulation.
In order to demonstrate the behavior, we implement the introduced energy-based controller in MATLAB. Table~\ref{Tab:Parametervalues} shows the used parameters for this experiment. 
The motor is assumed to have initial state $\{\theta_0,\dot{\theta}_0,I_0\}=\{0.1,5,0.4\}$ and the system is expected to converge to  the rest state $\{\dot{\theta}_d,I_d\}=\{0,0\}$ while the position arrives at its desired value $\theta_d=2$ rad.
\begin{table}[b!]
	\caption{Parameters of the DC motor and controller.}
	\label{Tab:Parametervalues}
	\centering
		\begin{tabular}{|cc|cc|}
			\hline
			Variable & Value & Variable & Value\\
			\hline
			$J$ & 0.004 kg-m$^2$ & $B$ & 0.04 Nm-s/rad \\
			$R$ & 1.3 \textnormal{$\varOmega$}& $L$ & 1 mH \\
			$K_b$ & 0.5 & $K_m$ & 0.5\\
			$K$  & 0.7 N/m & $k_P$ & 565 \\
			$k_{E,s}$ & 725 & $k_D$  & 0.07  \\
			$h_s$ & 0.11 & $u_{sat}$ & 45 V\\
			$K_c$& 610 & $K_L$& 0.4 \\
			\hline
	\end{tabular}
\end{table}
Based on the constraints for sampling period $h_k$ in (\ref{Eq:DeriLyapunovreorderedlowhintext}) and (\ref{Eq:singularfirstcondition})-(\ref{Eq:singularsecondcondition}), the random sampling range is chosen as $h_k\in[h_{min},h_{max}]=[0.05,0.2]$ s. It is important to note that in this arbitrary sampling for each iteration $h_k$, the number of times that sampling time changes is also random. The controller gains are also determined under the stability conditions with values given in Table~\ref{Tab:Parametervalues}.

  \begin{figure}[t!]
    \centering
    \vspace{3mm} 
        \includegraphics[width=2.2 in]{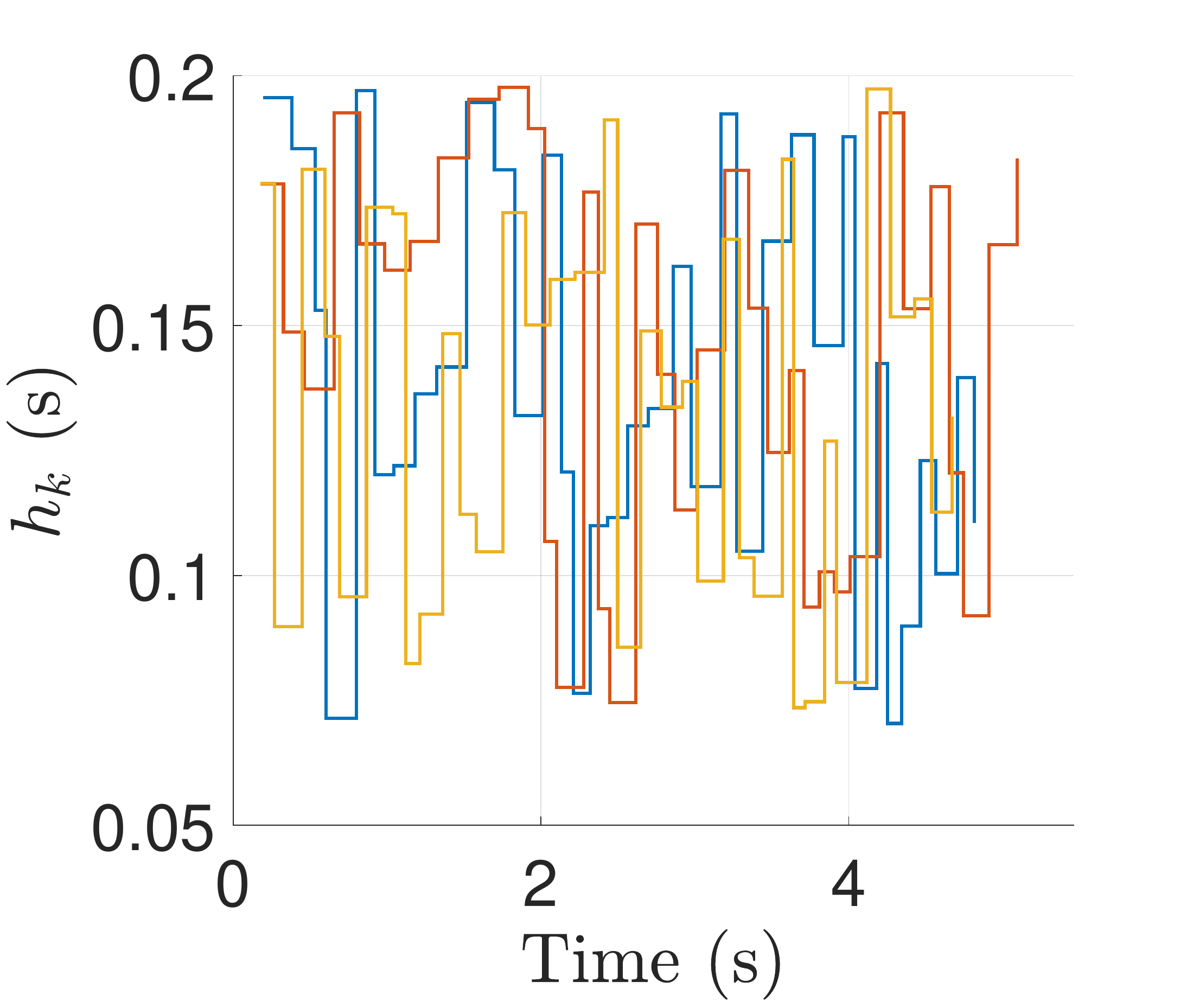}
	\caption{\makehighlight{Non-uniform random sampling.}} 
	\label{Fig:DiscreteResu4}
\end{figure}

\makehighlight{Figs.~\ref{Fig:DiscreteResu1}-\ref{Fig:DiscreteResu2} illustrate the simulation results for three different runs of random sampling. The developed discrete energy-based controller can successfully converge the system to desired states while sampling time is randomly altered for all cases (see Fig. \ref{Fig:DiscreteResu4} for details). }

Fig.~\ref{Fig:DiscreteResu2} presents how the energy is minimized towards the desired states convergence. \makehighlight{The same behavior is observed for the controller inputs without having any saturation.}

\begin{figure}[t!]
\centering
\vspace{3mm} 
    \includegraphics[width=2.6 in]{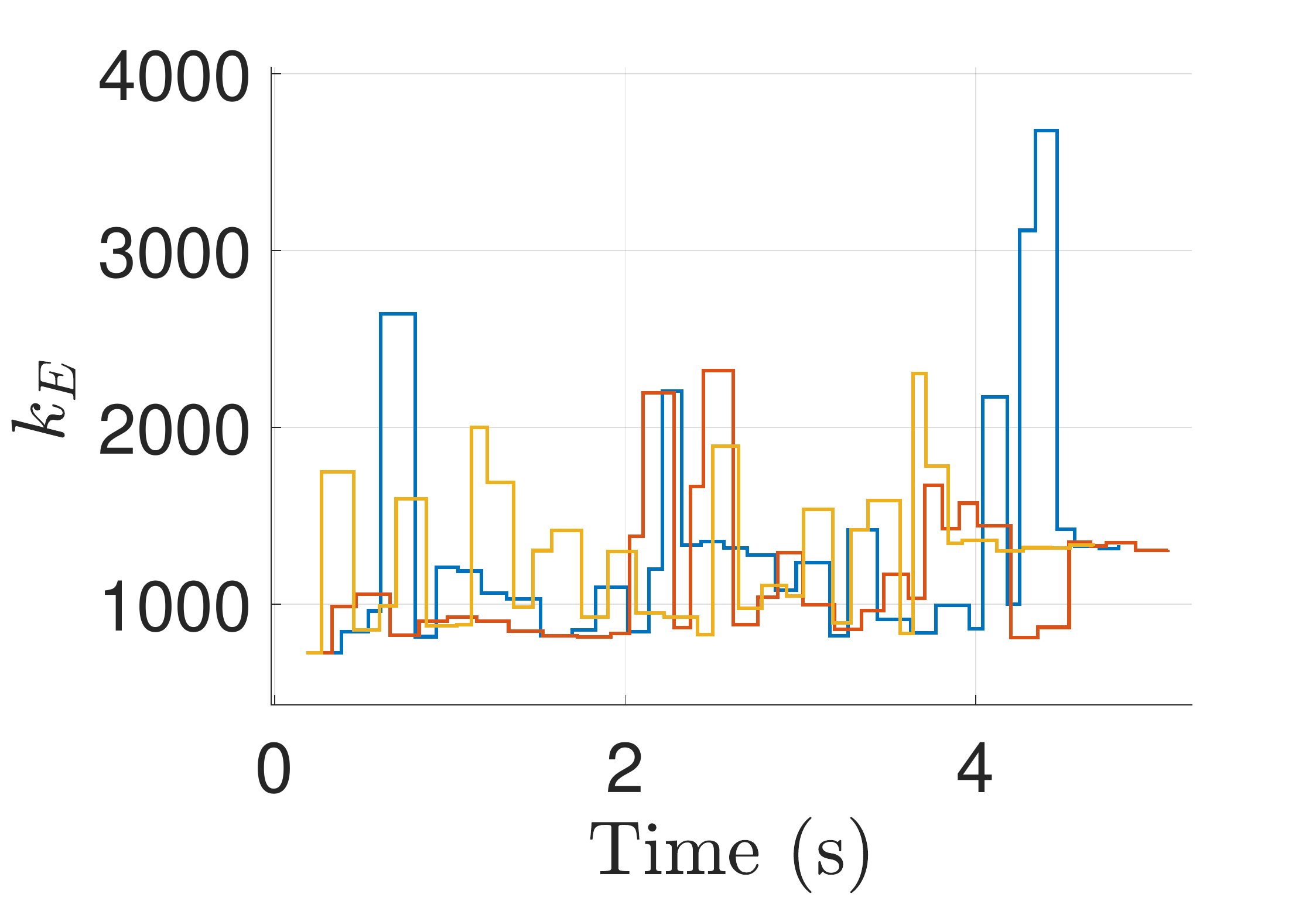}
\caption{\makehighlight{The controller $k_E$ gain tuning results.}} 
\label{Fig:DiscreteResu3}
\end{figure}

\begin{figure}[t!]
    \centering
    \vspace{3mm} 
        \includegraphics[width=2.8 in]{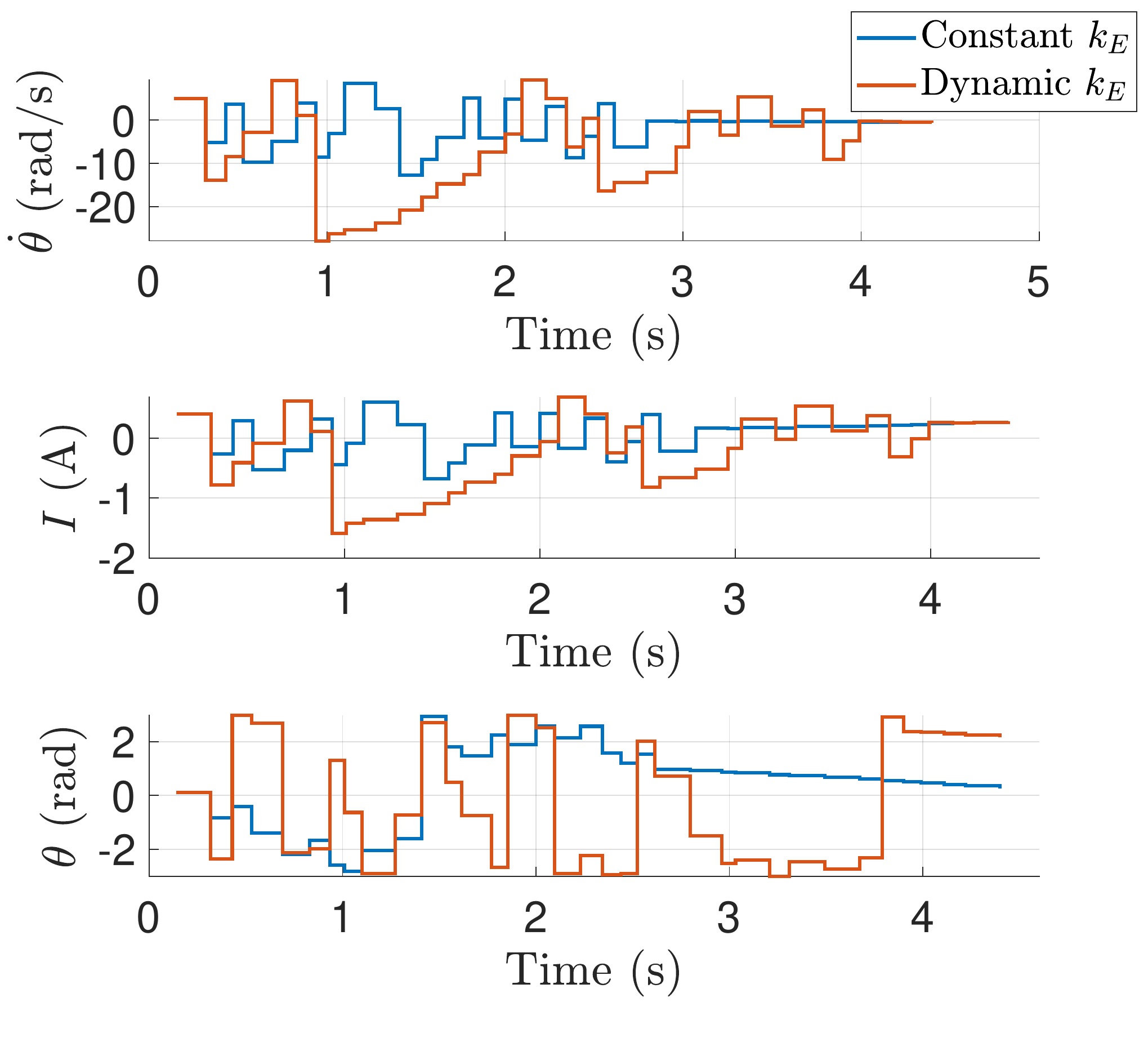}
        	\caption{\makehighlight{A comparison of constant energy gain and the proposed dynamic gain.}} 
	\label{Fig:DiscreteResu4Com}
\end{figure}

\begin{figure}[t!]
    \centering 
    \vspace{3mm} 
        \includegraphics[width=1.8 in]{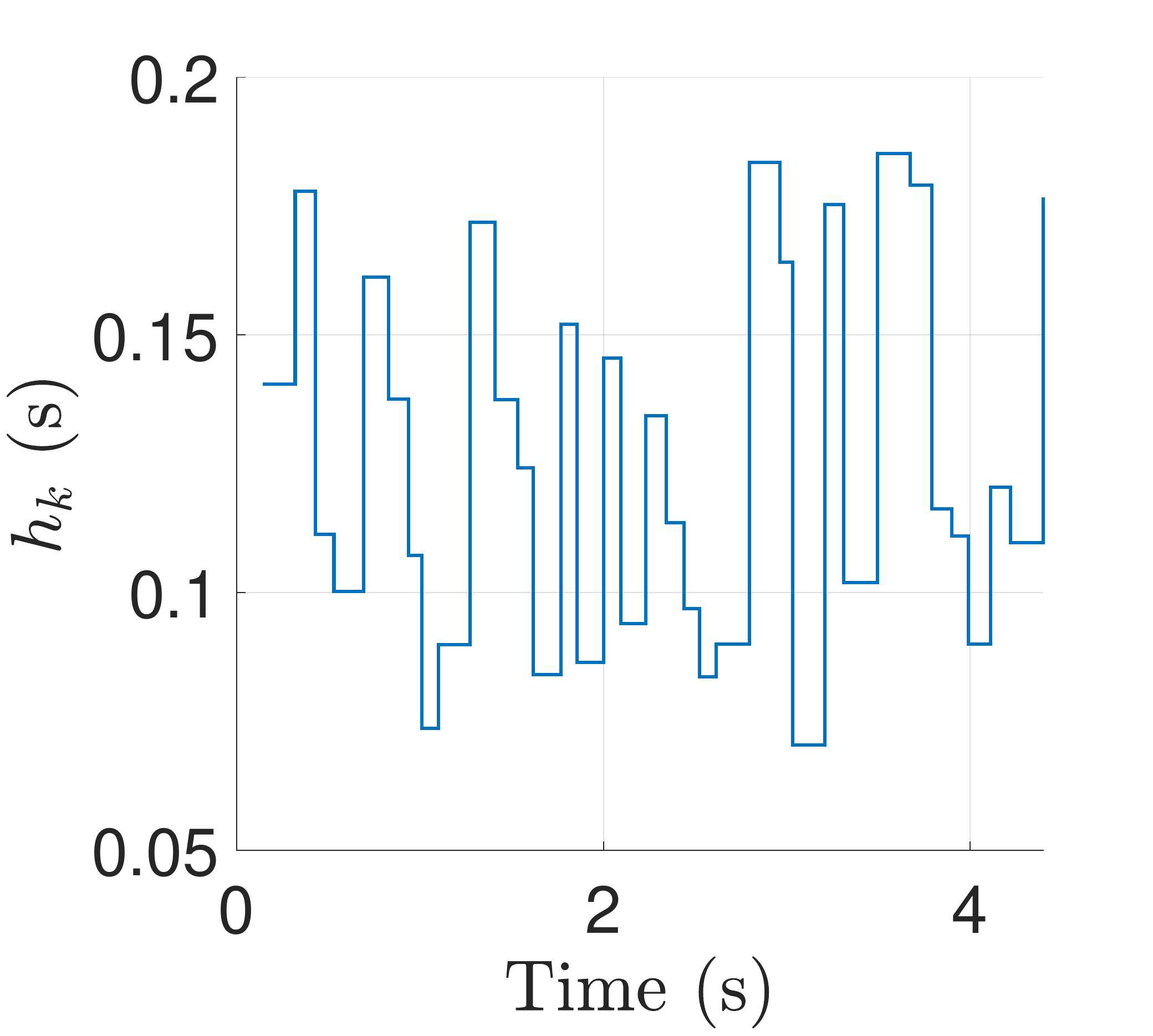}
         \includegraphics[width=1.9 in]{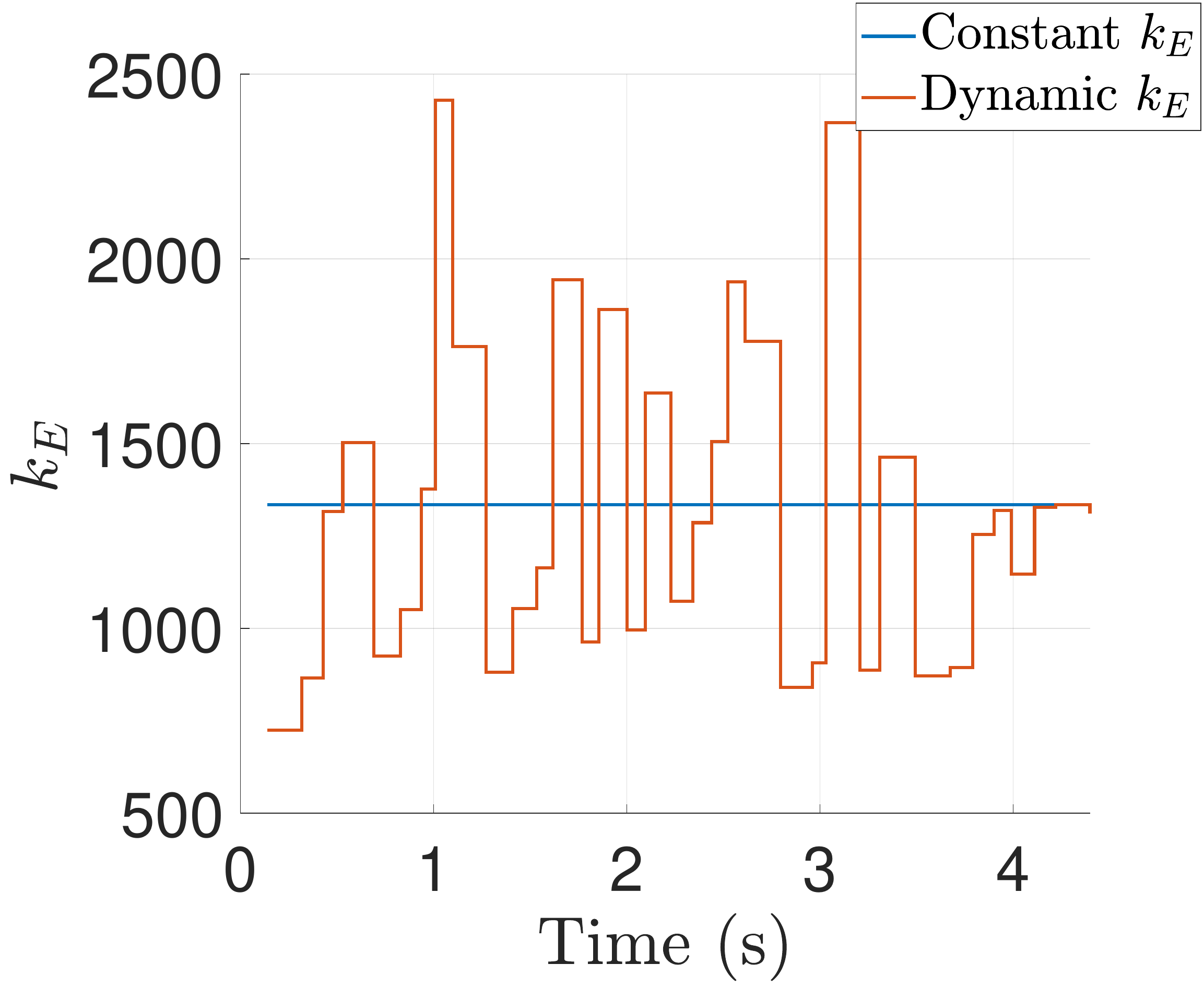}
        	\caption{\makehighlight{The non-uniform random sampling and controller gain $k_E$ results for constant and dynamic $k_E$.}} 
	\label{Fig:CompConVarySamcase513}
\end{figure}

Additionally, during the transition state, the controller is prone to have more fluctuations when the sampling time changes more often with large values. This can be observed clearly for the second random sampling case (the plot in red) where dynamic $k_E$ gain is getting re-tuned by (\ref{Eq:SamplingKeTune}) as a large value before $1 $ s (see Fig.~\ref{Fig:DiscreteResu3}). This large value of $k_E$ prevents the controller input from saturation which shows the developed controller works flawlessly.

To demonstrate that the controller with dynamic $k_E$ gain has important advantages, we make a comparison with a case where the gain $k_E$ is assigned a constant value. Fig.~\ref{Fig:DiscreteResu4Com} illustrates an example result of system convergence with random sampling time presented in Fig.~\ref{Fig:CompConVarySamcase513}. 
It is clear that our designed controller with dynamic gain $k_E$ can converge all the states to desired values but $k_E$ with constant gain fails to arrive at desired angular state $\theta_f$, which is 2 rad.

\section{Conclusion} \label{sec:conclusion}
In this paper, we have introduced a discrete energy-based controller that supports non-uniform digital control, which contributes to the research of control and scheduling co-design.
In this work, before developing the controller, the discrete state-space model was presented. A physical system (a DC motor) with elastic load and its model were introduced. Then, the energy-based controller with arbitrary switching was proposed where constraints under the stability condition of the overall Lyapunov function are found for controller gains. A new re-tuning function based on the sampling period is developed for the gain of the energy term to stabilize during the flexible switching. This work shows that the designed discrete energy-based controller can successfully converge the system into the desired states while the sampling time is randomly changing. 


In addition to what was presented in this paper, some further improvements can be made, including: 
(i) introduce a switch supervisor that chooses when and how to switch between multiple periods;
(ii) constrain the periods to a pre-defined set with a harmonic period to facilitate task scheduling and allocation; and (iii) explore other control strategies to extend the problem to digital systems without any physical properties, which was required in this energy-based controller. 
All of these will be form part of future work.
\section{Acknowledgment}
This work was supported by JSPS KAKENHI grant number JP21K20391 and partially by Japan Science and Technology Agency (JST) [Moonshot R$\&$D Program] under Grant JPMJMS2034. 

\bibliographystyle{IEEEtran}
\bibliography{references}

\end{document}